\documentclass[runningheads]{llncs}
\usepackage{graphicx}
\usepackage{algorithm}
\usepackage{algpseudocode}
\usepackage{amssymb,amsmath}
\usepackage{cleveref}
\usepackage{algpseudocode}
\usepackage{pdfpages}
\usepackage{placeins}
\usepackage{enumitem}

\begin{document}
\title{Point-Location in The Arrangement of Curves}
%
%
\author{Sepideh Aghamolaei\inst{1}\and
Mohammad Ghodsi\inst{1,2}}
\authorrunning{S. Aghamolaei, M. Ghodsi}
%
\institute{Department of Computer Engineering, Sharif University of Technology\\
\email{aghamolaei@ce.sharif.edu}\and
Institute for Research in Fundamental Sciences (IPM)\\
\email{ghodsi@sharif.edu}}
\maketitle              
\begin{abstract}
An arrangement of $n$ curves in the plane is given. The query is a point $q$ and the goal is to find the face of the arrangement that contains $q$. A data-structure for point-location, preprocesses the curves into a data structure of polynomial size in $n$, such that the queries can be answered in time polylogarithmic in $n$.

We design a data structure for solving the point location problem queries in $O(\log C(n)+\log S(n))$ time using $O(T(n)+S(n)\log(S(n)))$ preprocessing time, if a polygonal subdivision of total size $S(n)$, with cell complexity at most $C(n)$ can be computed in time $T(n)$, such that the order of the parts of the curves inside each cell has a monotone order with respect to at least one segment of the boundary of the cell. We call such a partitioning a curve-monotone polygonal subdivision.

\keywords{Point-location \and Pseudo-line Arrangement \and Data Structures \and Computational Geometry.}
\end{abstract}
\section{Introduction}
An arrangement of curves is a subdivision of the plane induced by those curves.
The arrangement of pseudo-lines has complexity $O(n^2)$, where a pseudo-line is a curve that goes to infinity in two directions, and two pseudo-lines can intersect at most once.
An output-sensitive algorithm for the arrangement of Jordan arcs with running time $O((n+k)\log n)$ exists, where $k$ is the number of intersections in the arrangement~\cite{toth2017handbook}.

Bézier curves based on Bernstein basis polynomials from computer graphics for non-rasterized curves are examples of real-world examples of polynomial curves.

Testing whether a given pseudoline arrangement can be converted into an straight-line arrangement is known as the stretchability problem is NP-hard and is equivalent to existential theory of the reals~\cite{shor1991stretchability}. A generalized version of this problem is the d-stretchability discusses the possibility of replacing each pseudoline with a polynomial of degree $d$ which requires $\Omega(\sqrt{n})$ for all simple arrangements~\cite{toth2017handbook}.

Finding a point in a planar subdivision and returning the face that contain that point is called {\em point location}.
Point location in a polygonal subdivision with $n$ segments can be solved in $O(n\log n)$ preprocessing to build the data-structure and $O(\log n)$ query time for $n$ segments~\cite{de1997computational}.
Without using a point-location data-structure, the problem can be solved in $O(nQ(n))$ time, for arbitrary shapes with an inclusion-exclusion oracle $Q(n)$, by testing whether the point is inside or outside each shape.

Since the problem of converting the arrangement of pseudo-lines into a polygonal arrangement is NP-hard~\cite{shor1991stretchability}, the point-location on polygonal subdivisions is not helpful.
The problem of point-location in a set of congruent disks has already been discussed, and an algorithm with $O(\log n+k)$ query time and $O(n^3\log n)$ preprocessing exists for it~\cite{aghamolaei2020windowing}.

\paragraph{Contributions}
The paper mainly generalizes the results of~\cite{aghamolaei2020windowing} to arbitrary curves, improving their query time and their parallelization.
\begin{itemize}
\item
We improve the previous work on the point-location of congruent disks to use $O(\log n)$ query time, instead of $O(\log n+k)$, at the cost of increasing the preprocessing time to $O(n^4\log n)$.
\item
We give a data structure for the point location problem in the arrangement of a set of curves with $O(T(n)+S(n)\log S(n))$ preprocessing and $O(\log n+\log C(n))$ query time, if an algorithm for constructing a planar subdivision of complexity $S(n)$ and cell size at most $C(n)$ in time $T(n)$ is given.
\end{itemize}
\section{The Point-Location Data-Structure for Curves}
In \Cref{alg:preprocess}, we use a polygonal partitioning such that the part of the curve inside each partition is monotone with respect to at least one boundary of that cell, which we call a {\em curve-monotone partitioning}. Then, we use a combination of point-location in polygonal subdivisions and binary search on the ordered set of polynomial arcs bounded by a segment of the subdivision (\Cref{alg:query}).

\subsection{A Curve-Monotone Polygonal Subdivision for Disjoint Arcs of Polynomials of Maximum Degree $2$}
After breaking the input into a set of disjoint $xy$-monotone convex/concave curves, \Cref{alg:subdivision} computes a subdivision of the plane, such that the order of the curves in each cell is monotone with respect to each cell boundary.

To do this, first compute the bounding box of each curve. Some of the bounding rectangles might be intersecting, or inside one another:
\begin{enumerate}[label=(\alph*)]
\item For a sequence of boxes inside one another (nested rectangles), merge their boxes by considering the outermost bounding box only, if there are no other rectangles outside this sequence inside them.
\item For intersecting boxes, add the tangents to the curves at the points of intersections with bounding rectangles of other curves and their own. Then, starting from a point on the boundary of the intersected area, move towards one of the curves, until hitting the other curve's bounding box or tangent, then repeat this on the other curve. If it is the intersection of two curves, draw the common tangent of the curves to separate them.
\item For disjoint rectangles inside another rectangle, we divide them by extending the edges of the rectangles inside the bigger rectangle, until intersecting the boundary of the extension of an edge of another rectangle.
\item For disjoint rectangles, extend their edges until intersecting the boundary or the extension of an edge to create a set of disjoint rectangles.
\item For intersecting nested bounding boxes, first handle the intersecting segments, then handle the disjoint and then the nested rectangles.
\end{enumerate}
An example of these cases is shown in~\Cref{fig:nested_intersecting}.
\begin{figure}[h]
	\centering
	\includegraphics[scale=0.7]{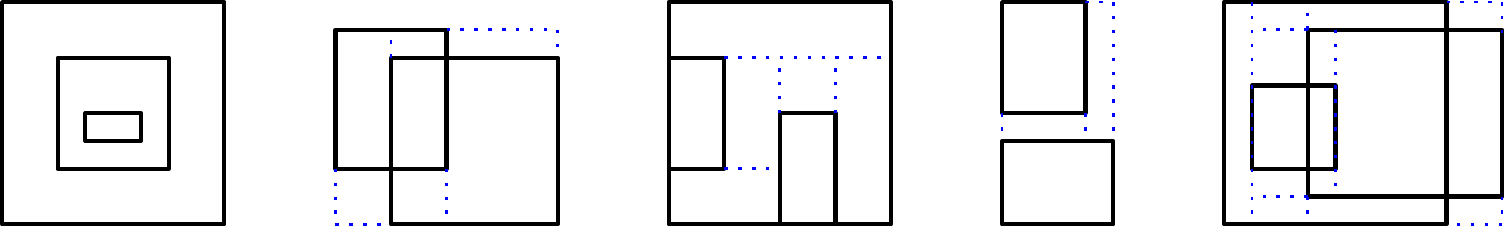}
	\caption{Four possible arrangements of bounding boxes and their subdivisions.}
	\label{fig:nested_intersecting}
\end{figure}

\begin{algorithm}[h]
\caption{Curve-Monotone Subdivision}
\label{alg:subdivision}
\begin{algorithmic}[1]
\Require{A set of curves $P_1,\cdots,P_n$}
\Ensure{A subdivision of the plane such that the distances of $\{P_i\}_{i=1}^n$ in each cell can be uniquely defined.}
\For{$i=1,\cdots,n$}
\State{$Q\gets Q\cup \{$Break each curve $P_i$ into $xy$-monotone totally concave or totally convex parts$\}$.}
\EndFor
\State{$A=$Build the arrangement of curves $Q$.}
\State{$B=$Connect the intersections and vertices using straight line segments.}
\State{$C=$Partition $B$ by drawing the bounding boxes of the curves in $Q$ or the mutual tangents of the curves in $Q$ when their bounding boxes intersect, except for the curves that are inside another curve.}
\State{Triangulate the cells of $C$ with complexity more than $3$.}
\State{Store for each segment $s$ the list of curves inside cell $c$ intersecting it as $L_{s,c}$ along with the index of their adjacent cells.}
\\ \Return{$C$ and $L_{s,c}, s\in c, c\in C$}
\end{algorithmic}
\end{algorithm}
\subsection{A Subdivision For Disks}
\Cref{alg:disks} computes a planar subdivision for a set of arbitrary disks.
\begin{algorithm}[h]
\caption{Subdivision for Disks}
\label{alg:disks}
\begin{algorithmic}[1]
\Require{A set of disks $P_1,\cdots,P_n$}
\Ensure{A subdivision of the plane $D$.}
\State{$A=$Build the arrangement of the disks $P_1,\cdots,P_n$.}
\State{$B=$Build the bounding box of the input, connect the intersections of the disks, their pairwise tangents and their centers to each other and to the boundary with straight lines.}
\State{$C=$Triangulate the cells of $B$ with complexity more than $3$.}
\State{Store for each segment $s$ the list of curves inside cell $c$ intersecting it as $L_{s,c}$ along with the index of their adjacent cells.}
\\ \Return{$C$ and $L_{s,c}, s\in c, c\in C$}
\end{algorithmic}
\end{algorithm}
An example of the possible cases of \Cref{alg:disks} is shown in \Cref{fig:disks}.
\begin{figure}[h]
	\centering
	\includegraphics[scale=0.7]{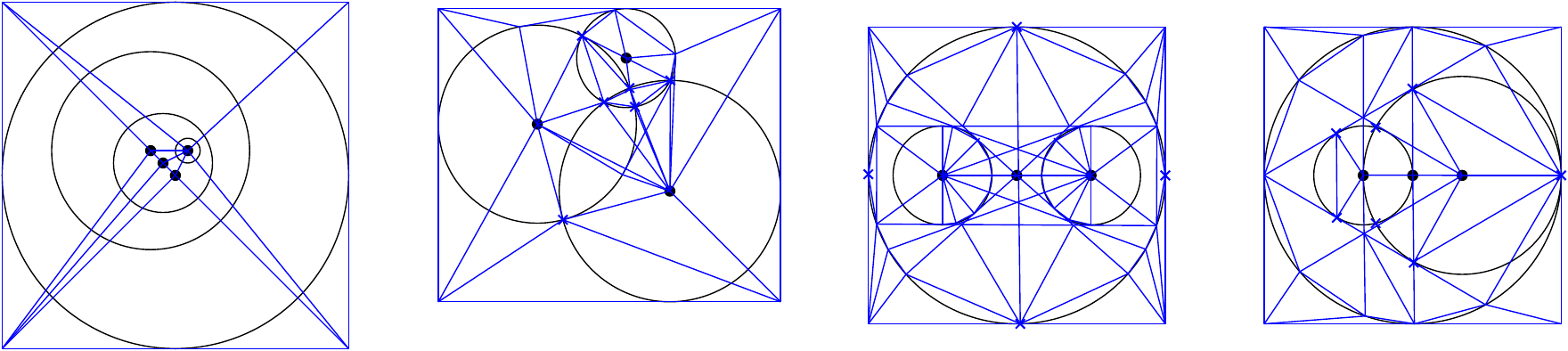}
	\caption{Four examples of the subdivisions of \Cref{alg:disks} for a set of disks.}
	\label{fig:disks}
\end{figure}
The size of the arrangement is $O(n^2)$, since at most $\binom{n}{2}$ lines connect the centers of the disks, there are $4\binom{n}{2}$ mutual tangents, and the bounding box adds at most $4$ vertices and therefore a multiple of at most $4$.
Because of the triangulation, the cell complexity is at most $3$
\subsection{The Point-Location Algorithm}
Assuming a curve-monotone planar subdivision is given, \Cref{alg:preprocess} builds a data-structure for the point-location problem.
\begin{algorithm}[h]
\caption{Preprocessing}
\label{alg:preprocess}
\begin{algorithmic}[1]
\Require{A set of curves $P_1,\cdots,P_n$}
\Ensure{A point location data structure $D$ augmented by a set of lists $L_{s,c}$}
\State{$C=$ Build a curve-monotone planar subdivision on $P_1,\cdots,P_n$.}
\State{Sort each list $L_{s,c}, s\in c, c\in C$, based on their distances from the segment $s$.}
\State{$D=$ Build a point-location data structure on $C$.}
\\ \Return{$D$ and $L_{s,c}, s\in c, c\in C$}
\end{algorithmic}
\end{algorithm}
Assuming the algorithm for constructing the curve-monotone planar subdivision has time complexity $T(n)$, space complexity $S(n)$ and maximum cell complexity $C(n)$, \Cref{alg:preprocess} takes $T(n)+O(S(n)\log S(n))$ for building the subdivision, sorting and constructing the point-location data-structure.

An example of a valid partitioning for the point-location using \Cref{alg:preprocess} is shown in \Cref{fig:partitioning_draft}.
\begin{figure}[h]
	\centering
	\includegraphics[scale=0.8]{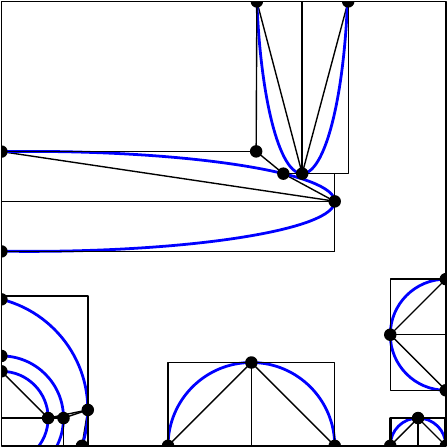}
	\caption{The input curves in blue, and the finer partitioning of $B$ into $C$ (\Cref{alg:preprocess}).}
	\label{fig:partitioning_draft}
\end{figure}

\begin{algorithm}[h]
\caption{Query}
\label{alg:query}
\begin{algorithmic}[1]
\Require{A point $q$, a point location data structure $D$ augmented by a set of lists $L_{s,c}$}
\Ensure{The cell of $D$ that contains $q$}
\State{$c=$ The output of point location of $q$ in $D$.}
\State{$X=\emptyset$}
\For{$s\in$ the boundary of $c$}
\State{$X\gets X\cup$ Binary search on the lists $L_{s,c}$, using the polynomial's equation to detect above/below to find the cell of the arrangement.}
\EndFor
\\ \Return{the only non-empty result in $X$.}
\end{algorithmic}
\end{algorithm}

\begin{theorem}
\Cref{alg:preprocess} using \Cref{alg:subdivision} for computing the subdivision has time complexity $O(n^4\log(n))$ and space complexity $O(n^4)$.
\end{theorem}
\begin{proof}
The sizes of intermediate structures is as follows:
\[
|Q|=O(n), |A|=O((n)^2), |B|=|A|, |C|\leq|B|^2=O(n^4),
\sum_{s\in c, c\in C}|L_{s,c}|=O(|C|)
\]
Triangulating the regions does not increase the overall complexity, since this adds a linear number of edges.

Breaking each polynomial into monotone parts takes $O(1)$ time.
Computing $D$ takes $O(|C|\log|C|)=O(n^4\log (n))$ time.
Triangulating the regions takes linear time, and sorting the set of segments in each cell takes $O(|C|\log|C|)$.
So, the time complexity is dominated by the time complexity of computing $D$.
\end{proof}

\begin{theorem}
\Cref{alg:query} has time complexity $O(\log n)$.
\end{theorem}
\begin{proof}
The point-location in $D$ takes $O(\log(n))$ time, since it is a polygonal point location for an input of size $O(n)$.
Based on \Cref{alg:preprocess}, the complexity of each cell is at most $4$, and each list has size at most $n$, so the binary searches take $O(\log(n))$ time.
So, the total time complexity of a query is $O(\log n)$.
\end{proof}
\section{Point Location via Nearest Neighbor: \\The Landmark Method}
A method of solving the point location problem is to use a set of points called landmark point-location~\cite{haran2006efficient,fogel2012cgal}, by computing the nearest neighbor to the query point among the landmarks. The landmarks used in previous work were the vertices of a $\sqrt{n}\times \sqrt{n}$ grid~\cite{haran2006efficient,fogel2012cgal}. After finding the nearest landmark, the exact cell that contains the point is computed via walking. The performance of these algorithms depended on the input.

For the arrangement of a set of equal disks, the set of centers has been used as landmarks, after being augmented by a sorted list of cells inside each disk, which are used in binary searches at query time~\cite{aghamolaei2020windowing}. While this method gives worst-case $O(\log n)$ guarantee on the query time, it does not directly generalize, even to disks of arbitrary sizes.

Choosing a set of points as landmarks can be used to solve the point-location problem. In other words, we want the current subdivision to be a subset of the Voronoi diagram of a point set (the landmarks). This can be done by placing two points per edge of the diagram, such that the perpendicular bisector of the line connecting them is the edge of the subdivision, and the points are the closest landmark points in the cells who share that edge. Store the information of neighboring cells of each point.
Unlike the landmark method, this refined version does not require walking, as each cell of the Voronoi diagram uniquely identifies a cell of the original subdivision.

Using this method, the number of points is $2|E|=O(n)$. So, the running time remains the same and the algorithm computes an exact solution.
By computing a balanced clustering of the landmarks points, a point-location data-structure that is a tree with degree $k+1$, where $k$ is the number of clusters. This reduces the time complexity of computing the solution from $O(\log n)$ to $O(\log_k n)$ in a parallel setting with $k$ processors. When processing a batch of queries, this improves upon sending the query points to all processors, which saves a factor $k$ in the communication at each level of the tree.
\section{Open Problems}
\paragraph{Approximate Point-Location.}
While building an exact point-location data-structure might be expensive in terms of time complexity, it is possible to build approximate point-location data structures by enclosing the original curve between two simpler curves. Then, the exact point-location data-structure needs to be used for the area between the enclosing simplified curves.
For the one-sided error, such as never reporting a point inside (outside) the shape as outside (inside), the data-structure built on the bounding (enclosed) curve is enough to get an approximate solution.
\paragraph{Application: Point-Location among Contour Lines.}
For a given objective function and its discretization, the contour lines separate the input into regions such that the value of the objective function is the same between two contour curves. By building a point-location data structure on such an arrangement, it is possible to get the values with error dependent on the discretization used when building the contour lines.
An example of this method exists for the length query problem~\cite{aghamolaei2020windowing}.
%
%
%
 \bibliographystyle{splncs04}
 \bibliography{refs}

\begin{thebibliography}{1}
\providecommand{\url}[1]{\texttt{#1}}
\providecommand{\urlprefix}{URL }
\providecommand{\doi}[1]{https://doi.org/#1}

\bibitem{aghamolaei2020windowing}
Aghamolaei, S., Keikha, V., Ghodsi, M., Mohades, A.: Windowing queries using
  {M}inkowski sum and their extension to {MapReduce}. Journal of Supercomputing
   (2020)

\bibitem{de1997computational}
De~Berg, M., Van~Kreveld, M., Overmars, M., Schwarzkopf, O.: Computational
  geometry. In: Computational geometry, pp. 1--17. Springer (1997)

\bibitem{fogel2012cgal}
Fogel, E., Halperin, D., Wein, R.: CGAL arrangements and their applications: A
  step-by-step guide, vol.~7. Springer Science \& Business Media (2012)

\bibitem{haran2006efficient}
Haran, I.: Efficient point location in general planar subdivisions using
  landmarks. Tel Aviv University (2006)

\bibitem{shor1991stretchability}
Shor, P.: Stretchability of pseudolines is {NP}-hard. Applied Geometry and
  Discrete Mathematics-The Victor Klee Festschrift  (1991)

\bibitem{toth2017handbook}
Toth, C.D., O'Rourke, J., Goodman, J.E.: Handbook of discrete and computational
  geometry. Chapman and Hall/CRC (2017)

\end{thebibliography}
\end{document}